\documentclass[runningheads,11pt]{llncs}
\usepackage{hyperref}
\usepackage[T1]{fontenc}

\usepackage{booktabs}
\usepackage{authblk}
\usepackage{fullpage}
\usepackage{amsmath,amssymb}
\usepackage{breqn}
\usepackage{todonotes}
\presetkeys{todonotes}{inline}{}
\usepackage{xspace} 
\usepackage{complexity} 
\usepackage{thm-restate}
\usepackage{amssymb}
\usepackage{cleveref} 
\usepackage{url}
\usepackage{dsfont} 
\usepackage{subcaption}

\usepackage{lineno,setspace}
\usepackage{algpseudocode}
\usepackage{algorithm}

\usepackage{tcolorbox,bbm}

\usepackage{tikz}
	\usetikzlibrary{calc,fit,arrows.meta}
    \usetikzlibrary{decorations.pathreplacing} 
\everypar{\color{black}} 

\allowdisplaybreaks

\newcommand{\ins}{\mathcal N}       % generic instance
\newcommand{\iW}{\mathcal W}      %instance A
\newcommand{\iC}{\mathcal C}  
\usepackage{geometry}
\usepackage{setspace}
\usepackage{array}
\usepackage{adjustbox}

\begin{document}

\title{Effects of Vote Delegation in Blockchains: Who Wins?}

\author{}
\institute{}
\author{Hans Gersbach \inst{1} \and
Manvir Schneider\inst{2}\and
Parnian Shahkar\inst{3}}
\authorrunning{F. Author et al.}
\institute{KOF Swiss Economic Institute, ETH Zurich, Switzerland \email{hgersbach@ethz.ch}
 \and
Cardano Foundation, Switzerland \email{manvir.schneider@cardanofoundation.org} \and 
University of California, Irvine, USA
\email{shahkarp@uci.edu} }

\maketitle

\begin{abstract}
Should blockchain allow for vote delegation? This paper investigates which alternative benefits from vote delegation in binary collective decisions within blockchains. We begin by examining two extreme cases of voting weight distributions: \textit{Equal-Weight} (\textit{EW}), where each voter has equal voting weight, and \textit{Dominant-Weight} (\textit{DW}), where a single voter holds a majority of the voting weights before any delegation occurs. We show that vote delegation tends to benefit the ex-ante minority under \textit{EW}, i.e., the alternative with a lower initial probability of winning. The converse holds under \textit{DW} distribution. Through numerical simulations, we extend our findings to arbitrary voting weight distributions, showing that vote delegation benefits the ex-ante majority when it leads to a more balanced distribution of voting weights. Finally, in large communities where all agents have equal voting weight, vote delegation has a negligible impact on the outcome. As a practical consequence, vote delegation can be beneficial for blockchains with highly unbalanced voting rights, but not for those with balanced rights. In decentralized finance (DeFi), vote delegation is widely adopted to streamline governance and increase participation. However, it remains unclear when delegation actually aligns outcomes with community preferences\footnote{This research was partially supported by the Zurich Information Security and Privacy Center (ZISC) at ETH Zurich. We thank Roger Wattenhofer, Huseyin Yildirim, Marcus Pivato, Florian Brandl and participants of the Astute Modeling Seminar at ETH Zurich for their valuable feedback. }.
\end{abstract}

\section{Introduction}

Staking is a fundamental tool of Proof-of-Stake (PoS) blockchains for enhancing their chain and guaranteeing economic security. Examples include Cardano\footnote{\url{https://www.cardano.org} (retrieved January 20, 2022). See also~\cite{Karakostas2020}.}, Solana~\cite{Yakovenko2017}, Polkadot~\cite{Wood2016}, Tezos~\cite{Goodman2014}, and Concordium\footnote{\url{https://www.concordium.com} (retrieved June 7, 2022).}. These blockchains allow agents to delegate their stakes to other agents for validation purposes or to govern the blockchain. Typically, the agents who delegate, referred to as {\em delegators}, do not know the preferences of the agents to whom they delegate, as every participant is represented merely by an address in the form of a number or a pseudonym. This model captures real-world delegation systems in DeFi protocols and DAOs, such as Arbitrum, Compound, or Optimism, where voters often interact pseudonymously and delegate voting power without complete information. Our analysis is directly applicable to on-chain governance in DeFi, where vote delegation is common. Protocols such as Uniswap and Optimism rely on delegated voting to manage proposals and upgrades. By identifying how delegation affects outcomes under different token distributions, our findings offer design insights for such systems, especially when participation is low or token holdings are concentrated.\\

A weighted voting system is a method of decision-making in which agents are allocated a number of votes or a voting weight that varies according to specific criteria. Unlike traditional voting systems where each voter has an equal vote, the influence or power of each voter in a weighted voting system is proportional to their assigned voting weight. In blockchains, particularly in the new generation of Proof-of-Stake blockchains, voting weights can vary among agents, often tied to their stakes in the system. Throughout this paper, whenever we use the term "voting", we are specifically referring to weighted voting. Also, we use "weights" and "voting weights" interchangeably.\\

The agents involved in the governing body of a blockchain or an electronic voting system are divided into two distinct groups: voters and delegators. Delegators are those agents who prefer not to participate directly in the voting process. In a conventional voting system, delegation is prohibited, and thus, these individuals would abstain. However, if vote delegation is permitted, each delegator delegates their votes to a voter. Since delegators do not know the preferences of the voters, all voters appear alike to them. Consequently, each voter has an equal chance of receiving votes from a delegator.
There are various reasons why individuals might choose to delegate their votes. These reasons include avoiding the costs associated with becoming informed about the issues and alternatives, or gaining additional financial returns, as is often the case with staking on blockchains. 
Conversely, voters are the agents who always participate in the voting process. In a conventional voting system, the voting weight of these agents is determined solely by their own stakes. However, when delegation is allowed, their voting weight is augmented by any additional weight they receive from delegators. This dynamic can significantly influence the overall voting outcome, as voters who received more votes gain more influence in the decision-making process. \\

We address the following problem: In a blockchain or an electronic voting system, a governing body must choose between two alternatives, denoted as {\em A} and {\em B}. These alternatives may represent voting on a node software upgrade, with {\em A} for the upgrade and {\em B} for staying with the current software. A node upgrade aims to enhance the blockchain's capabilities. Thus, the majority supports the upgrade, while the minority, due to infrastructure requirements, opposes it. 
Alternatively, the decision could involve invalidating ({\em A}) or validating ({\em B}) an incorrect transaction. The underlying assumption in PoS blockchains is that a majority will invalidate incorrect transactions, while a minority of possibly malicious agents may attempt to validate them to disrupt the system. Transaction validation requires agents to run a node and blockchain software. Agents unable or unwilling to run a node can delegate their voting weights to other agents involved in the validation process. In blockchains, the preferred alternative, which ensures the integrity and security of the network, is expected to be supported by the majority. In this example, the preferred choice is invalidating incorrect transactions. Our goal is to explore instances where delegation enhances the probability of the preferred alternative winning, thereby enabling delegation under such conditions. We employ a random model where each voter independently votes for the preferred alternative, denoted as alternative {\em A}, with probability $p$ ($0.5 < p < 1$). We categorize individuals favoring {\em A} ({\em B}) as {\em A-voters} ({\em B-voters}). The probability that alternative {\em A} wins always exceeds $\frac{1}{2}$, designating {\em A} as the {\em ex-ante majority}, and {\em B} as the {\em ex-ante minority}.\\
%%%%%%%

We analyze two extreme distributions of voter weights: an \textit{Equal-Weight} distribution ({\em EW}), where voting weights are evenly distributed, and a \textit{Dominant-Weight} distribution ({\em DW}), where one voter holds a majority of the weights. Throughout the paper, we will use these abbreviations, {\em EW} and {\em DW}, to refer to these distributions. Unless otherwise specified, no assumptions are made regarding the voting weight distribution of delegators. 

Our study reveals four key findings:
First, if the weight distribution of voters before delegation is {\em EW}, delegation favors alternative {\em B}, whereas if the distribution is {\em DW}, it benefits \emph{A}.
These outcomes, derived from comprehensive proofs, demonstrate the critical influence of voters' weight distribution in determining the favored side under delegation. 
Second, when all delegators have equal weight, the probability of {\em A} winning converges to that under the {\em EW} distribution as the number of delegators increases. This result holds regardless of the initial distribution of voter weights.
Third, if the weight distribution of voters is {\em EW}, the probability that {\em A} wins approaches one as the number of voters increases. This result holds true even if any arbitrary number of equal-weight delegators delegate their votes, rendering the effect of delegation negligible.
Fourth, numerical analysis illustrates how balance in the distribution of voters' weights influences the likelihood of {\em A}'s victory. A more balanced weight distribution of voters increases the probability of {\em A} winning. These insights underscore that the initial distribution of voters' weight is crucial in determining which alternative potentially benefits from delegation. \\

Our findings have several practical implications. In small, balanced communities, even a single instance of vote delegation can significantly reduce the probability of the preferred alternative winning, suggesting that delegation should be avoided. In large, equal-weight communities, delegation has little impact, so other factors like financial incentives should guide the decision to allow it. Conversely, in blockchains with highly unbalanced voting weights, delegation helps balance the weights and improves the preferred alternative's chances. Additionally, in scenarios with many equal-weight delegators, delegation can shift the weight distribution towards equal-weight, making it beneficial for the blockchain community. Our findings offer design insights for DAO governance structures, especially in systems where token distributions are skewed, and delegation may either mitigate or increase concentration of voting power.

Our results also speak to the consequences of Sybil attacks, i.e. the possible incentives for agents to create multiple identities to which other agents can delegate. We will show that in case of having a Sybil attacker, we get to the \textit{DW} regime, and thereby the likelihood of the minority winning is the highest. Therefore, in these cases, vote delegation should be strictly prohibited.  \\

This paper is organized as follows. In Section \ref{sec: literature}, we discuss the related literature. Section \ref{sec: model} introduces our model. In Sections \ref{sec: results1} and \ref{sec: results2}, we study the impact of delegation where the weight distribution of voters is {\em DW} and {\em EW} respectively. Section \ref{sec:asymptotic} outlines the asymptotic behavior of delegation. In Section \ref{sec:numerical} we present our numerical experiments. In Section \ref{sybill-attack}, we discuss ways to prevent Sybil attacks that are already used in practice and we discuss how results might be affected if such an attack happened nevertheless. Finally, Section \ref{sec: conclusion} concludes the paper.

\section{Related Literature} \label{sec: literature}

Vote delegation has attracted significant attention in democratic literature, particularly within computational social choice theory and more recently in the blockchains. In computational social choice, this practice is often termed liquid democracy.\\

Gersbach et al.~\cite{GMM21} find that vote delegation often results in more favorable outcomes with higher probability compared to the conventional voting in costly voting environments, when malicious voters are present, and when preferences are private. Gersbach et al.~\cite{Gersbach2022} review how vote delegation, under both free and capped scenarios, may impact voting outcomes with private preference information. Conversely, Kahng et al.~\cite{liquid_algorithmic} examine a scenario where voters have different levels of information. In their model with information acquisition, they show that even delegations from less-informed to better-informed voters, may reduce the likelihood of choosing the preferred alternative. In our paper, we investigate how the distribution of voter weights affects the probability of each alternative winning the election, when preferences are private. This assumption on private information is in line with the studies by \cite{GMM21} and \cite{Gersbach2022}. Our model builds on the model of \cite{Gersbach2022} and allows us to study how different voters' weight distributions impact delegation. Caragiannis~and~Micha~\cite{Caragiannis2019} extend the work of \cite{liquid_algorithmic}, demonstrating that liquid democracy can result in less desirable outcomes compared to conventional voting.\\

Further literature on liquid democracy includes works by Christoff and Grossi~\cite{christoff2017binary}, and Brill and Talmon~\cite{brill2018pairwise}. The latter paper assumes that voters can delegate votes to better-informed voters that reflect their partial order over preferences. \cite{christoff2017binary} analyzes delegation cycles and how delegations affect individual rationality in liquid democracy. Our model does not allow for delegation cycles as the set of delegators is separated from the set of voters.\\

Weidener et al.~\cite{Weidener2025} conduct a comprehensive review of vote delegation in DAOs, with a focus on transparency, centralization, and voter passivity. Their work highlights the need for better theoretical tools to understand when delegation improves outcomes—precisely what our model addresses. Bongaerts et al.~\cite{bongaerts2025votedelegationdefigovernance} provide empirical evidence from multiple DeFi protocols showing how vote delegation shapes governance outcomes and concentrates influence among a few key actors. Strnad~\cite{strnad2025delegationparticipationdecentralizedgovernance} applies an epistemic view to decentralized governance, demonstrating that while partial abstention can enhance decision accuracy, explicit vote delegation often undermines it.\\

Voting power (Banzhaf voting index) within the context of liquid democracy is examined by Zhang~and~Grossi~\cite{Zhang_Grossi_2021}. In Proof-of-Stake blockchains, agents holding a stake of the native cryptocurrency, can delegate their stake to other agents, for example, to agents that run a stake pool. Delegation to stake pools when preferences of stake pool owners (and delegators) is private information is analyzed in Gersbach et al.~\cite{gersbach2022staking} and Schneider~\cite{schneider2023staking}. By delegating to a stake pool, agents expect some staking reward which can be translated to gaining utility if the preferred alternative is implemented in a liquid democracy setting. The problem of selecting an appropriate stake pool (or delegatee) when agents have different levels of information about (or trust in) other agents is studied by Zhang et al.~\cite{zhang2023rationally}.\\

Bloembergen et al.~\cite{rational_delegation} and Escoffier et al.~\cite{escoffier-gilbert} study the strategic behavior of delegation on networks. The former identifies conditions for the existence of Nash equilibria in such games, while the latter demonstrates that in more general setups, no Nash equilibrium may exist, and it may even be $NP$-complete to decide whether one exists at all. \cite{liquid_algorithmic} analyze the potential benefits of delegating votes to neighbors in a network structure. \\

% Recent work on representative democracy includes studies by \cite{PIVATO202052} and \cite{ijcai2019-1}. The first paper examines a model where legislators are chosen by voters, and votes are counted according to the number of votes a representative owns, and they show that the voting outcome in a large election is the same as in conventional voting, where delegation is not allowed. \cite{SohVoutsa} extends this model by incorporating other voting forms, such as weighted approval voting. \cite{ijcai2019-1} studies flexible representative democracy, where experts are elected first, and then voters either delegate their vote to the representatives or vote themselves. In our model, delegators delegate their votes uniformly at random to any voter in the society. In flexible representative democracy, random uniform delegation can be achieved by uniformly allocating votes among the representatives.\\
Amanatidis et al.~\cite{amanatidis2023potential} study vote delegation for the on-chain governance of the Cardano blockchain, also known as CIP-1694.\footnote{See \url{https://www.1694.io/en} (retrieved June 25, 2024).} In particular, they examine elections where voters must express preferences on numerous issues, often resulting in incomplete preferences due to the vast number of alternatives. This study, motivated by blockchain governance systems (in particular, Cardano), investigates whether delegating votes to knowledgeable proxies can enhance election outcomes, identifying conditions for socially better results through theoretical and experimental analysis.

\section{The Model}\label{sec: model}

We consider a large polity (a society or a blockchain community) that faces a binary choice between {\em A} and {\em B}. There is a group of $m \in \mathbb{N}_+$ individuals (the \textit{delegators}) who do not want to vote and either abstain under conventional voting or delegate their votes if vote delegation is allowed. The remaining population (\textit{voters}) votes. Voters have private information about their preference for {\em A} or {\em B}. A voter prefers alternative {\em A} ({\em B}) with probability $p$ ($1-p$), where $0<p<1$. Without loss of generality, we assume $p>\frac{1}{2}$, designating {\em A} as the {\em ex-ante majority}. Voters favoring {\em A} ({\em B}) are called {\em A}-voters ({\em B}-voters).
The assumption that preferences are private information implies that all voters are alike for delegators and there is no way for delegators to extract information and learn about voters' preferences. Consequently, delegation is performed uniformly at random, that is, every voter has equal chances to receive votes through delegation. \\

Let $\ins = [j]^n_{j=1}$ be a set of $n \geq 2$ voters. Voters have initial weights $w=[w_j]^n_{j=1}$ and delegators have initial weights $D=[d_j]^m_{j=1}$. The weight of an individual represents the multiplicity of their vote. Typically, in blockchain systems, a voter's stake dictates their voting weights. The decision rule is a simple majority weighted voting. 
We compare two voting processes:
\begin{itemize}
    \item Conventional voting: vote delegation is not allowed; thereby delegators abstain and each voter $j$ casts votes equivalent to their weight $w_j$.
    \item Post-delegation voting: vote delegation is allowed. Each of the $m$ delegators delegate their votes to one voter uniformly at random, disregarding the current weights of the voters. Consequently, voters with both lesser and greater initial weights are equally likely to be selected. The total weight of any voter who acquires additional votes through delegation increases in direct proportion to the number of votes delegated to them. In post-delegation voting, voters have weights $w^d=[w^d_j]^n_{j=1}$. Now, each voter $j$ casts votes proportional to their new weight $w^d_j$. Consequently, $\sum_{i=1}^n w_i^d = \sum_{j=1}^n w_j +\sum^{m}_{k=1} d_k.$
\end{itemize}

The power set of $\ins$ is denoted as $2^\ins$. We define a function $\Omega: 2^\ins, \mathbbm{N}^n \to \mathbbm{N}$, where for any set of voters $S\in 2^\ins$ and any vector of voters' weights $\omega$, $\Omega(S, \omega) = \sum_{j \in S} \omega_j$. We denote $P(A \mid \omega)$ as the probability that {\em A} wins, when voters' weights are given by $\omega$. $P(A \mid \omega)$ is computed by the following equation:

\begin{equation}\label{pr_conv}
    P(A \mid \omega) = \sum_{S \in 2^\ins} p^{|S|}(1-p)^{n-|S|}G(S, \omega), 
\end{equation}

where \( p^{|S|}(1-p)^{n-|S|} \) is the probability that exactly the members of \( S \) are \emph{A}-voters, and \( G(S, \omega) \) is the probability that \emph{A} wins given that only the members of \( S \) are \emph{A}-voters. We compute $G(S, \omega)$ as follows: 
\begin{align*}\label{Gformula}
  G(S, \omega):=\begin{cases}
               1, & \text{if } \Omega(S,\omega)>\frac{\Omega(\ins,\omega)}{2}, \\
               \frac{1}{2}, & \text{if } \Omega(S,\omega)=\frac{\Omega(\ins,\omega)}{2},\\
               0, & \text{if } \Omega(S,\omega)<\frac{\Omega(\ins,\omega)}{2}.
            \end{cases}    
\end{align*}

The following example illustrates the defined notations.

\begin{example}\label{example1}
Consider a voting where $\ins = \{1,2,3\}$ represents a set of voters, with corresponding weights $w = (1,2,4)$. Assume that each voter decides whether to vote for {\em A} independently with a probability $p$. The table \ref{tab:my_label} enumerates all possible sets of {\em A-} voters, and the probability of {\em A} winning in each case. 
\begin{table}[h]
\centering
\begin{tabular}{|c|c|c|}
\hline
\( S \) & Prob. \( S \) supporting {\em A} & Prob. {\em A} wins with support of \( S \) \\ \hline
\( \varnothing \) & \( (1-p)^3 \) & 0 \\ \cline{1-3}
\( \{1\} \) & \( p(1-p)^2 \) & 0 \\ \cline{1-3}
\( \{2\} \) & \( p(1-p)^2 \) & 0 \\ \cline{1-3}
\( \{3\} \) & \( p(1-p)^2 \) & 1 \\ \cline{1-3}
\( \{1,2\} \) & \( p^2(1-p) \) & 0 \\ \cline{1-3}
\( \{1,3\} \) & \( p^2(1-p) \) & 1 \\ \cline{1-3}
\( \{2,3\} \) & \( p^2(1-p) \) & 1 \\ \cline{1-3}
\( \{1,2,3\} \) & \( p^3 \) & 1 \\ \hline
\end{tabular}
\caption{Probability of {\em A} winning for various subsets of voters where $w= (1,2,4)$.}
\label{tab:my_label}
\end{table}
It is obvious that {\em A} wins only if the set of {\em A}-voters includes the third voter, as this voter holds the majority of the weights. The probability of {\em A} winning is computed as follows:

$$P(A \mid w) = p(1-p)^2 + 2p^2(1-p) + p^3 = p.$$
\end{example}
We study the effect of vote delegation on the probability that {\em A} wins, when the voters' weight distribution in the conventional voting ($w$) is either of the following:
\begin{itemize}
    \item {\em Equal-Weight (EW)}: all voters have equal weights. Therefore, $w$ is a vector of all ones. Throughout this paper, we denote this distribution with the subscript {\em EW}.

    \item {\em Dominant-Weight (DW)}: one of the voters has more than half of the total weights of voters. Without loss of generality, let $n$ be such a voter, where $w_n > \frac{\Omega(\ins,w)}{2}$. We denote $n$ as the \textit{dominant voter}. \footnote{In Example \ref{example1}, the third voter is the dominant voter.} Throughout this paper, we denote this distribution with the subscript {\em DW}.

\end{itemize}

A delegator's weight indicates the number of votes they transfer to a voter. In our analysis, unless specified differently, we do not assume specific weights for the delegators. This means our results are applicable regardless of the weight distribution of delegators. Let $\iW^d$ be the collection of all possible post-delegation weight vectors of voters. The probability of {\em A} winning in the conventional voting is denoted by $P^c(A)$. The probability of {\em A} winning after \( m \) delegators with weight distribution \( D \) have delegated their votes is denoted by \( P^d(A \mid D) \), and is computed by the following equation:

\begin{equation}\label{pr_deleg_dictator}
    P^d(A\mid D) = \sum_{\omega\in\iW^d} P(w^d = \omega) P(A \mid \omega), 
\end{equation}
where $P(w^d = \omega)$ is the probability that the voters' weight vector is $\omega$ post-delegation, and $P(A \mid \omega)$ is the probability of {\em A} winning given the post-delegation weight vector $\omega$.

\section{Delegation under the \textit{Dominant-Weight} distribution}\label{sec: results1}

In this section, we explore how delegation affects the likelihood of {\em A} winning given the voters' initial weights follow a {\em DW} distribution. In the conventional voting, there exists a voter, known as the dominant voter, who owns the majority of the votes. Initially, it is uncertain which alternative the dominant voter will support. This voter, like every other voter, casts their vote for {\em A} with a probability of $p$. We denote the chance of {\em A} winning by conventional voting as $ P_{DW}^c(A)$, where $DW$ shows that a dominant voter exists. During the delegation process, each delegator assigns their votes to a randomly chosen voter, without taking the voters' weights into account. This implies that low-weight and high-weight voters have an equal chance of being chosen. As a result of delegation, the dominant voter may no longer be the only decisive voter. The process thus allows other voters to influence the outcome in favor of {\em A}, thereby increasing {\em A}'s winning probability, as we will demonstrate later in this section. The probability of {\em A} winning post-delegation, given that the voters' initial weights follow a {\em DW} distribution, is denoted as \( P_{DW}^d(A \mid D) \).

\begin{example}\label{example2}
Building on Example \ref{example1}, assume the initial voter weights vector is $w = (1, 2, 4)$. From our previous calculation, we know that the probability of {\em A} winning in the conventional voting is $P(A \mid w) = p$. Now, we introduce two delegators, each holding a single vote. After delegating their votes, the possible weight configurations of the voters may be any of the vectors within $\iW^d = \{(3, 2, 4), (1, 4, 4), (1, 2, 6), (2, 3, 4), (1, 3, 5), (2, 2, 5)\}$. It can be easily verified that for all these new weight vectors $w^d \in \iW^d$, $P(A \mid w^d) \geq P(A \mid w)$ when $p > 0.5$. Thus, the likelihood of {\em A} winning post-delegation is either unchanged or increased compared to the conventional voting.
\end{example}
This example demonstrates that delegation can dilute the concentration of power, reduce the dependence on a single dominant voter's decision, and thereby increase {\em A}'s chances of winning.

The following theorem demonstrates that the probability of \(A\) winning is lowest when a dominant voter is present. As vote delegation can disrupt this dominance, it benefits {\em A}, the ex-ante majority. 

\begin{theorem}\label{thm1}
$P_{DW}^d(A\mid D) \geq  P_{DW}^c(A)$ for any $m\geq 1$ and $p>0.5$.
\end{theorem}
\begin{proof}%[Theorem \ref{thm1}]
Without loss of generality, let us assume that the $n$th voter is the dominant voter. We denote $\ins\setminus\{n\}$ as the set of all voters, except for the $n$th voter. The collection of all subsets of $\ins\setminus\{n\}$ is denoted as $2^{\ins\setminus\{n\}}$. From equation~\eqref{pr_conv}, it can be observed that in the conventional voting, {\em A} wins if and only if the dominant voter votes for her/him. In other words, any subset of {\em A}-voters including $n$-th voter results in {\em A} winning. Therefore, the probability of {\em A} winning in the conventional voting, denoted as $ P_{DW}^c(A)$, is computed by the following equation: 

\begin{equation}\label{pr_conv_dict}
     P_{DW}^c(A) = \sum_{S \in 2^{\ins\setminus\{n\}}} p^{|S|+1}(1-p)^{n-|S|-1} = p.
\end{equation}

For the post-delegation voting, let us define $\iW^d_1 = \{\omega \in \iW^d : \exists j \in \ins, \omega^d_j > \Omega(\ins,\omega)/2 \}$. In other words, $\iW^d_1$ is the set of all post-delegation weight vectors that contain a dominant voter. Conversely, $\iW^d_2$ is defined as the complement of $\iW^d_1$ within the set $\iW^d$, representing all post-delegation weights that do not contain a dominant voter. Hence, we can rewrite equation~\eqref{pr_deleg_dictator} as the following: 
\begin{equation}\label{pr_deleg_dictator0}
    P_{DW}^d(A\mid D) = \sum_{\omega\in\iW^d_1} P(w^d = \omega) P(A \mid \omega) + \sum_{\omega\in\iW^d_2} P(w^d = \omega) P(A \mid \omega), 
\end{equation}

For $\omega \in \iW^d_1$, it is easy to observe $P(A \mid \omega) = P_{DW}^c(A) = p$. Consequently, the expression simplifies to:
\begin{equation}\label{pr_deleg_dictator1}
    P_{DW}^d(A\mid D) = \sum_{\omega\in\iW^d_1} P(w^d = \omega) p + \sum_{\omega\in\iW^d_2} P(w^d = \omega) P(A \mid \omega).
\end{equation}
Recall equation~\eqref{pr_conv} for $\omega \in \iW^d_2$:
    \begin{equation}
    P(A \mid \omega) = \sum_{S \in 2^\ins} p^{|S|}(1-p)^{n-|S|}G(S, \omega).
    \end{equation}

    For any set $S \in 2^\ins$, either it does not include the $n$-th voter, the dominant voter in the conventional voting, and thus $S \in 2^{\ins \setminus \{n\}}$, or it includes voter $n$ and therefore $S\setminus \{n\} \in 2^{\ins \setminus \{n\}}$. Therefore, the probability can be expressed as:
        
    \begin{align}\label{pr_new_express}
    \begin{split}
    P(A \mid \omega) = \sum_{S \in 2^{\ins\setminus \{n\}}} &[p^{|S|+1}(1-p)^{n-|S|-1}G(S \cup \{n\}, \omega) \\ & + p^{|S|}(1-p)^{n-|S|}G(S, \omega)],
    \end{split}
    \end{align}    

    where the first term in summation refers to subsets including the $n$-th voter, and the second term refers to subsets excluding her/him. \\
    The difference $P(A \mid \omega) -  P_{DW}^c(A)$ equals to
    \begin{equation}
        \sum_{S \in 2^{\ins\setminus \{n\}}} [p^{|S|+1}(1-p)^{n-|S|-1}(G(S \cup \{n\}, \omega) - 1) + p^{|S|}(1-p)^{n-|S|}G(S, \omega)].
    \end{equation}

    Let us define $\iC_1 = \{S \in 2^{\ins\setminus \{n\}} : G(S \cup \{n\}, \omega) = 0 \}$, $\iC_2 = \{S \in 2^{\ins\setminus \{n\}} : G(S, \omega) =  1 \}$, $\iC_3 = \{S \in 2^{\ins\setminus \{n\}}: G(S \cup \{n\}, \omega) = 1/2 \}$, and $\iC_4 = \{S \in 2^{\ins\setminus \{n\}} : G(S, \omega) =  1/2 \}$. Put simply, $\iC_1$ is the collection of {\em A}-voters such that even if voter $n$ is added to them, {\em A} loses. $\iC_2$ is the collection of {\em A}-voters, without voter $n$, that will result in the victory of {\em A}. $\iC_3$ is the collection of {\em A}-voters such that by the addition of voter $n$, there will be a tie. $\iC_4$ is the collection of {\em A}-voters, excluding voter $n$, that will result in a tie. Using these new notions we obtain
    \begin{equation}\label{c1c2}
        \begin{aligned}
            & P(A \mid \omega) -  P_{DW}^c(A) = \\ & \underbrace{\sum_{S \in \iC_1} -p^{|S|+1}(1-p)^{n-|S|-1} + \sum_{S\in \iC_2} p^{|S|}(1-p)^{n-|S|}}_{\Delta_1} \\ & + \frac{1}{2}\underbrace{[\sum_{S \in \iC_3} -p^{|S|+1}(1-p)^{n-|S|-1} + \sum_{S\in \iC_4} p^{|S|}(1-p)^{n-|S|}]}_{\Delta_2}.
        \end{aligned}
    \end{equation}

By inspecting the definitions of $\iC_1$ and $\iC_2$, it can be observed that for any $S\in \iC_1$, the complement of $S \cup \{n\}$ within the set $\ins$, represented as $(S \cup \{n\})^c$, is a unique set in $\iC_2$. Similarly, for any $S\in \iC_2$, the complement of $S$ excluding $n$, represented as $S^c \setminus \{n\}$, is a unique set in $\iC_1$. The one-to-one correspondence between sets $\iC_1$ and $\iC_2$ allows us to write every set $S \in \iC_1$ as ${S'}^c\setminus \{n\}$ for a unique $S' \in \iC_2$. Therefore, we can write the following:

\begin{equation}
    \Delta_1 = \sum_{S\in \iC_2} -p^{n-|S|}(1-p)^{|S|} + p^{|S|}(1-p)^{n - |S|}.
\end{equation}

Sets in $\iC_2$ can be grouped based on their size, therefore by defining $\iC^j_2 = \{S\in \iC_2: |S| = j\}$, we have that $\iC_2 = \bigcup^{n-1}_{j=1} \iC^j_2$. Hence we obtain
\begin{equation}\label{eq:long}
    \begin{aligned}
     \Delta_1 &= \sum^{n-1}_{j=1} |\iC^j_2|  [-p^{n-j}(1-p)^{j} + p^{j}(1-p)^{n - j}] \\
    &=\sum^{\lfloor n/2 \rfloor}_{j=1} |\iC^j_2| [-p^{n-j}(1-p)^{j} + p^{j}(1-p)^{n - j}] \\ & \quad + 
    \sum^{n-1}_{j=\lfloor n/2 \rfloor + 1} |\iC^j_2| [-p^{n-j}(1-p)^{j} + p^{j}(1-p)^{n - j}] \\
    &= \sum^{\lfloor n/2 \rfloor}_{j=1} |\iC^j_2| [-p^{n-j}(1-p)^{j} + p^{j}(1-p)^{n - j}] \\ & \quad + |\iC^{n-j}_2| [- p^{j}(1-p)^{n - j}+p^{n-j}(1-p)^{j}] \\
    &\geq \sum^{\lfloor n/2 \rfloor}_{j=1} |\iC^j_2|[-p^{n-j}(1-p)^{j} + p^{j}(1-p)^{n - j} \\ & \quad - p^{j}(1-p)^{n - j}+p^{n-j}(1-p)^{j}] \\
    &= 0.
\end{aligned}
\end{equation}

The first equality arises because $\iC_2 = \bigcup_{j=1}^{n-1} \iC^j_2$. The second equality follows by decomposing the sum. The third equality results from a change of variable in the second summation. The inequality is derived from the fact that $|\iC^{n-j}_2| \geq |\iC^j_2|$ for all $j \in [1, \lfloor n/2 \rfloor]$, noting that if $j = n/2$, the corresponding term in the brackets is zero. Finally, the last equality holds because each term in the brackets is zero.

So far, we have shown that $\Delta_1 \geq 0$. The steps for showing $\Delta_2 \geq 0$ are very similar. First, one can show that there exists a one-to-one correspondence between the sets $\iC_3$ and $\iC_4$. Therefore, every set $S \in \iC_3$ can be written as ${S'}^c \setminus \{n\}$ for a unique $S' \in \iC_4$, allowing us to write $\Delta_2$ as follows:
\begin{equation}
    \Delta_2 = \sum_{S \in \iC_4} -p^{n - |S|}(1 - p)^{|S|} + p^{|S|}(1 - p)^{n - |S|}.
\end{equation}
Similar to what we did for $\iC_2$, we can group the sets in $\iC_4$ based on their size. Defining $\iC^j_4 = \{S \in \iC_4 : |S| = j\}$, we have that $\iC_4 = \bigcup_{j=1}^{n-1} \iC^j_4$. Taking a closer look at $\Delta_2$, we can write a set of arguments analogous to those in equation \eqref{eq:long} by replacing $\iC^j_2$ with $\iC^j_4$, and show that $\Delta_2 \geq 0$. 
Finally, since $\Delta_1 \geq 0$ and $\Delta_2 \geq 0$, from equation \eqref{c1c2} it is evident that $P(A \mid \omega) \geq P_{DW}^c(A)$, For all $\omega \in \iW^d_2$, any $m\geq 1$ and $p>0.5$. By combining this result with equations \eqref{pr_conv_dict} and \eqref{pr_deleg_dictator1}, we obtain $P_{DW}^d(A\mid D) \geq  P_{DW}^c(A)$ and this concludes the proof of Theorem \ref{thm1}.
\end{proof}

\section{Delegation under the {\textit{Equal-Weight}} distribution}\label{sec: results2}

In this section, we examine the impact of vote delegation on the probability of {\em A} winning given that the voters' initial weights follow an {\em EW} distribution. Under this assumption, the probability of {\em A} winning in conventional voting is denoted as \( P_{EW}^c(A) \), and the probability of {\em A} winning post-delegation is denoted as $P_{EW}^d(A\mid D)$. The following theorem demonstrates that when all voters have equal weight, the probability of \(A\) winning is maximized. Since vote delegation can disrupt this equality, it undermines the ex-ante majority and thereby benefits the ex-ante minority.

\begin{theorem}\label{thm_balance}
$P_{EW}^c(A) \geq P_{EW}^d(A\mid D)$ for any $m\geq 1$ and $p>0.5$.
\end{theorem}

\begin{proof}
    
To present voters' initial weights, define $\mathbf{1}_n$ as an $n$-dimensional vector of ones. In this setting, $G(S, \mathbf{1}_n)$ equals 1 if $|S| > \frac{n}{2}$, equals $0$ if $|S| < \frac{n}{2}$, and equals $\frac{1}{2}$ if  $|S| = \frac{n}{2}$. The probability of {\em A} winning in the conventional voting denoted as $ P_{EW}^c(A)$ is computed by the following equation:
\begin{equation}\label{p_balance}
 P_{EW}^c(A) = \sum_{S \in 2^\ins} p^{|S|} (1-p)^{n-|S|} G(S, \mathbf{1}_n). 
\end{equation}

The probability of {\em A} winning post-delegation, denoted as $P_{EW}^d(A\mid D)$ is obtained from equation~\eqref{pr_deleg_dictator}, where $\iW^d$ is the collection of all possible post-delegation weights of voters. Since $\sum_{\omega\in\iW^d} P(w^d = \omega) = 1$, to prove Theorem~\ref{thm_balance}, it suffices to show that $P_{EW}^c(A) \geq P(A \mid \omega)$ for all $\omega \in \iW^d$. According to the definition of $P(A \mid \omega)$ in \ref{pr_conv}, we can write the following: 
\begin{equation}\label{diff_balance}
     P_{EW}^c(A) - P(A \mid \omega) = \sum_{S \in 2^\ins} p^{|S|} (1-p)^{n-|S|} [G(S, \mathbf{1}_n) - G(S, \omega)].
\end{equation}

For any $\omega \in \iW^d$, define the following collections for subsets of voters: \\
$\iC_1 = \{S \in 2^\ins: |S| < \frac{n}{2} \land \Omega(S,\omega) > \frac{\Omega(\ins,\omega)}{2}\}$
$\iC_2 = \{S \in 2^\ins: |S| > \frac{n}{2} \land \Omega(S,\omega) < \frac{\Omega(\ins,\omega)}{2}\}$, $\iC_3 = \{S \in 2^\ins: |S| = \frac{n}{2} \land \Omega(S,\omega) > \frac{\Omega(\ins,\omega)}{2}\}$, $\iC_4 = \{S \in 2^\ins: |S| = \frac{n}{2} \land \Omega(S,\omega) < \frac{\Omega(\ins,\omega)}{2}\}$
and $\iC_5 =  \{S \in 2^\ins: \Omega(S,\omega) = \frac{\Omega(\ins,\omega)}{2}\}$.
In essence, $\iC_1$ contains subset of voters who, despite being fewer than half of the total voters, possess sufficient collective voting weight to dictate the election outcome by supporting a specific alternative. Conversely, $\iC_2$ contains subsets of voters who outnumber the rest, yet, their collective weight is less than the remaining voters. $\iC_3$ consists of subsets of {\em A}-voters that include half of the voters and lead to {\em A} winning. Conversely, $\iC_4$ comprises subsets of {\em A}-voters that include half of the voters and result in {\em A} losing. Additionally, $\iC_4$ includes subsets of {\em A}-voters that lead to a tie.

The terms in the summand of \eqref{diff_balance} can be non-zero only if $S \in \iC_1 \cup \iC_2 \cup \iC_3 \cup \iC_4 \cup \iC_5$. Therefore, this equation can be rewritten as follows:
\begin{equation}\label{eq:diffEW}
\begin{aligned}
         P_{EW}^c(A) - P(A \mid \omega) =& \underbrace{\sum_{S \in \iC_2} p^{|S|} (1-p)^{n-|S|} - \sum_{S \in \iC_1} p^{|S|} (1-p)^{n-|S|}}_{\Delta_1} \\
         & + \underbrace{\sum_{S \in \iC_3\cup \iC_4} p^{\frac{n}{2}} (1-p)^{\frac{n}{2}} [\frac{1}{2} - G(S, \omega)]}_{\Delta_2} \\ & + \underbrace{\sum_{S \in \iC_5} p^{|S|} (1-p)^{n-|S|} [G(S, \mathbf{1}_n) - \frac{1}{2}]}_{\Delta_3}.
         \end{aligned}
\end{equation}
We analyze $\Delta_1$, $\Delta_2$, and $\Delta_3$ separately.
By inspecting the definitions of $\iC_1$ and $\iC_2$, we observe that the complement set of every set in $\iC_1$ is in $\iC_2$, and the complement set of every set in $\iC_2$ is in $\iC_1$. Therefore we can write every set $S\in \iC_2$ as the complement of a set $S^c\in \iC_1$, and obtain: 
\begin{equation}
\begin{aligned}
     & \Delta_1 = \sum_{S \in \iC_1}  p^{n - |S|} (1-p)^{|S|} -p^{|S|} (1-p)^{n-|S|} 
     \end{aligned}
\end{equation}
Since all sets $S \in \iC_1$, contain less than half of the voters, and $p > 0.5$, each term in the above summand is non-negative. Hence $\Delta_1 \geq 0$. 

Observe that for every $S\in \iC_3$, $S^c\in \iC_4$, and for every $S\in \iC_4$, $S^c\in \iC_3$. Therefore, 
\begin{equation}
\begin{aligned}
    \Delta_2 = &\sum_{S \in \iC_3} p^{\frac{n}{2}} (1-p)^{\frac{n}{2}} [\frac{1}{2} - G(S, \omega)] + p^{\frac{n}{2}} (1-p)^{\frac{n}{2}} [\frac{1}{2} - G(S^c, \omega)] = \\& \sum_{S\in \iC_3} -\frac{1}{2}p^{\frac{n}{2}} (1-p)^{\frac{n}{2}} + \frac{1}{2}p^{\frac{n}{2}} (1-p)^{\frac{n}{2}} = 0.
\end{aligned}
\end{equation}
we can group the sets in $\iC_5$ based on their size. Defining $\iC^j_5 = \{S \in \iC_5 : |S| = j\}$, we have that $\iC_5 = \bigcup_{j=1}^{n-1} \iC^j_5$. Now we can write $\Delta_3$ as following:
\begin{equation}
\begin{aligned}
        \Delta_3 =& \sum^{n-1}_{j=1} |\iC^j_5| p^{j} (1-p)^{n-j} [G(S, \mathbf{1}_n) - \frac{1}{2}] \\
        & = \sum^{\lceil \frac{n}{2}\rceil -1}_{j=1} - \frac{1}{2}|\iC^j_5| p^{j} (1-p)^{n-j} + 
        \sum^{n}_{j=\lfloor \frac{n}{2}\rfloor +1} \frac{1}{2}|\iC^j_5| p^{j} (1-p)^{n-j} \\
        & = \sum^{\lceil \frac{n}{2}\rceil -1}_{j=1} - \frac{1}{2}|\iC^j_5| p^{j} (1-p)^{n-j} + \frac{1}{2}|\iC^{n-j}_5| p^{n-j} (1-p)^{j},
\end{aligned}
\end{equation}
where the first equality follows by decomposing the sum, and the second equality results from a change of variable in the second summation. Now observe that for every set $S\in \iC^j_5$, $S^c\in \iC^{n-j}_5$, and for every set $S\in \iC^{n-j}_5$, $S^c\in \iC^j_5$. Therefore, for every $j$, $|\iC^j_5| = |\iC^{n-j}_5|$. Hence,
\begin{equation}
    \begin{aligned}
        \Delta_3 = \sum^{\lceil \frac{n}{2}\rceil -1}_{j=1} \frac{1}{2}|\iC^j_5| [p^{n-j} (1-p)^{j} - p^{j} (1-p)^{n-j}].
    \end{aligned}
\end{equation}
Given that $p > 0.5$ and $j < \frac{n}{2}$ in all terms of the summand, each term is non-negative. Consequently, this ensures that $\Delta_3 \geq 0$.
Now combining $\Delta_1 \geq 0$, $\Delta_2 = 0$, and $\Delta_3 \geq 0$ with equation \eqref{eq:diffEW}, $P_{EW}^c(A) - P(A \mid \omega) \geq 0$, for all $\omega \in \iW^d$.
\end{proof}

\section{Asymptotic Behavior of Delegation}\label{sec:asymptotic}

In this section, we study the asymptotic behavior of delegation under three different assumptions. In the first, we show that if the probability that voters vote for {\em A} goes to 1, the probability of {\em A} winning goes to 1 as well. In the second case, we demonstrate that when all delegators have equal weights and their number approaches infinity, the weight distribution of voters converges to the {\em EW} distribution. Notably, this result holds regardless of the initial weight distribution of the voters. However, if the delegators have an arbitrary weight distribution, this convergence does not occur. Finally, we prove that when all voters have equal weight, the probability of \( A \) winning converges to 1 as the number of voters approaches infinity. Moreover, this result holds even in the presence of any number of same-weight delegators.

First, we obtain a rather straightforward proposition:

\begin{proposition}\label{prop: lim_p_to_1}
For any fixed natural number $m \geq 1$, we have $\lim_{p\rightarrow 1}P^d(A\mid D,p)=1$. 
\end{proposition}

\begin{proof}
We can verify that $P^d(A\mid D,p)=1$ when $p=1$. 
\begin{align*}
    &\lim_{p\rightarrow 1} P^d(A\mid D,p)  \\ &=    
     \sum_{\omega\in\iW^d} P(w^d = \omega) \lim_{p\rightarrow 1}
    \underbrace{[\sum_{S \in 2^\ins} p^{|S|}(1-p)^{n-|S|}G(S,\omega)]}_{=1} = 1.
\end{align*}
Indeed, only the term $S= \ins$ survives.
% Since the function $P^d(A\mid D)$ is uniformly continuous in $p$, the Proposition is proved.
\end{proof}
% Proposition~\ref{prop: lim_p_to_1} shows that the probability that {\em A} wins approaches 1 when the probability that a voter is an {\em A}-voter goes to 1. 

Next, we show that when the number of voters is odd, the probability that {\em A} wins post-delegation converges to the probability that {\em A} wins when voters' weights follow an {\em EW} distribution, given delegators have equal weights, and their number approaches infinity. To present delegators' weights, define $\mathbf{1}_m$ as an $m$-dimensional vector of all ones. The probability of {\em A} winning in the {\em EW} model is the same as $P_{EW}^{c}(A)$.

\begin{proposition}\label{P_lim}
$\lim_{m\rightarrow \infty}P^d(A\mid \mathbf{1}_m)= P_{EW}^c(A)$ for any $p\in(0,1)$.
\end{proposition}
\begin{proof}
Probability of {\em A} winning when voters' distribution is {\em EW} is obtained in equation~\eqref{p_balance}. Let $w = [w_j]_{j=1}^{n}$ be the vector representing the initial weights of the voters. Since all delegators have equal weight, the probability of {\em A} winning post-delegation can be written as the following: 
\begin{equation}\label{prob_deleg_equal}
    P_{EW}^d(A\mid \mathbf{1}_m) = \sum_{S \in 2^\ins} p^{|S|}(1-p)^{n-|S|} G(S,m),
\end{equation}
where $G(S,m)$ is the probability of {\em A} winning if only voters in $S$ vote for A, while $m$ delegators with same weight have delegated their vote. $G(S,m)$ is obtained from the following:
\begin{equation}
    G(S,m) = \sum^m_{h=0} \binom{m}{h} (\frac{|S|}{n})^h (1-\frac{|S|}{n})^{m-h} g(S, h, m),
\end{equation}
where $g(S, h, m)$ is the probability that {\em A} wins if $h$ votes are delegated to the voters in $S$, and only those voters vote for A. We compute $g(S, h, m)$ as follows: 
\begin{align*}
  g(S, h, m) :=\begin{cases}
               1, & \text{if } \sum_{j\in S} w_j + h >\frac{\sum^n_{j=1}w_j + m}{2}, \\
               \frac{1}{2}, & \text{if } \sum_{j\in S} w_j + h = \frac{\sum^n_{j=1}w_j + m}{2},\\
               0, & \text{if } \sum_{j\in S} w_j + h < \frac{\sum^n_{j=1}w_j + m}{2}.
            \end{cases}    
\end{align*}
To prove the Proposition, we need to show that $lim_{m\rightarrow \infty} G(S,m) \to G(S, \mathbf{1}_n)$ for any fixed set of voters $S\in 2^\ins$.     Now, let $a = \frac{\sum_{j \notin S} w_j -\sum_{j\in S} w_j
 + m}{2}$ and $q = \frac{|S|}{n}$. Using the new notations, $G(S,m)$ is as follows:
    \begin{equation}
        G(S,m) = \sum^m_{h=0} \binom{m}{h} q^h (1-q)^{m-h} g(S, h, m),
    \end{equation}
    where given $g(S,h,m)$, 
    \begin{equation}
        \sum^m_{h=\lceil a \rceil+1} \binom{m}{h} q^h (1-q)^{m-h} \leq G(S,m) \leq  \sum^m_{h=\lceil a \rceil} \binom{m}{h} q^h (1-q)^{m-h}.
    \end{equation}
    The left summand above is equal to the probability $P[X \geq \lceil a \rceil +1]$
    for a random variable $X \sim Bin(m, q)$. Also, the right summand is equal to the probability $P[X \geq \lceil a \rceil ]$ for the same random variable. Hence, 
    \begin{equation}
        P[X \geq \lceil a \rceil +1] \leq G(S,m) \leq P[X \geq a].
    \end{equation}
    We use Hoeffding’s inequality to bound $G(S,m)$ for two different cases: 
\begin{itemize}
    \item if $q < \frac{1}{2}$:\\
    Since $|S| < \frac{n}{2}$, $G(S, \mathbf{1}_n) = 0$. \\
    Hoeffding's Inequality yields: 
\begin{equation*}
    P[X \geq m (q+\epsilon)] \leq \exp(-2\epsilon^2 m).
\end{equation*}
We compute $\epsilon$ by solving $m(q+\epsilon) = a$:
\begin{align*}
 \epsilon = \frac{\sum_{j\notin S} w_j  -\sum_{j\in S} w_j +m}{2m}-q.
\end{align*}

Let $t:=\sum_{j\notin S} w_j -\sum_{j\in S} w_j$, then by Hoeffding's Inequality:
\begin{align*}
    P[X \geq a] &\leq \exp(-2 (\frac{\sum_{j\notin S} w_j  -\sum_{j\in S} w_j +m}{2m}-q)^2 m)\\
    &= \exp(-m(2q^2-2q+\frac{1}{2}) -t-\frac{t^2}{2m}+2qt).
\end{align*}
The right-hand side (RHS) goes to $0$ for $m \to \infty$ if $2q^2-2q+\frac{1}{2}>0$. This inequality is true for any $q \neq \frac{1}{2}$. As $q<\frac{1}{2}$, the RHS goes to $0$ for  $m \to \infty$. That is,
\begin{equation*}
    \lim_{ m \to \infty } P[X \geq a] = 0.
\end{equation*}
Hence, for $q < \frac{1}{2}$, we have $\lim_{m \to \infty}G(S,m) =G(S, \mathbf{1}_n) = 0$.

\item if $q > \frac{1}{2}$:\\
    Since $|S| > \frac{n}{2}$, $G(S, \mathbf{1}_n) = 1$. \\
    Hoeffding's Inequality yields: 
    $$P[X < m (q-\epsilon)] \leq \exp(-2\epsilon^2 m).$$
We find $\epsilon$ by solving 
$m(q-\epsilon) = a+2$: 
$$ \epsilon = q - \frac{\sum_{j\notin S} w_j -\sum_{j\in S} w_j+m+2}{2m}.$$

Let $t:=\sum_{j\notin S} w_j  -\sum_{j\in S} w_j$, then, by Hoeffding's Inequality:
\begin{align*}
    P[X < \lceil a \rceil +1] & \leq P[X < a+2] \\ & \leq \exp(-2 (q-\frac{\sum_{j\notin S} w_j  -\sum_{j\in S} w_j+m+2}{2m})^2 m)\\
    &= \exp(-m(2q^2-2q+\frac{1}{2}) 
    -\frac{(t+4)^2}{2m} + (2q-1)(t+4)
    ).
\end{align*}
The RHS of the latter converges to $0$ for $m \to \infty$ if $2q^2-2q+\frac{1}{2}>0$. This inequality is true for any $q \neq \frac{1}{2}$. As $q>\frac{1}{2}$, the RHS converges to $0$ for  $m \to \infty$. That is, $\lim_{ m \to \infty } P[X < \lceil a \rceil +1] = 0.$
Therefore, $\lim_{ m \to \infty } P[X \geq \lceil a \rceil +1] = \lim_{ m \to \infty } (1-P[X < \lceil a \rceil +1]) = 1.$ Hence, for $q > \frac{1}{2}$, we have $\lim_{m \to \infty} G(S,m) = G(S, \mathbf{1}_n) = 1$

\end{itemize}
Since the number of voters is odd, $q \neq \frac{1}{2}$ and therefore the above analysis shows that $\lim_{m \to \infty} G(S,m) = G(S, \mathbf{1}_n)$ for any fixed set of voters $S$.    
\end{proof}

The following example shows that Proposition \ref{P_lim} fails if the delegators have arbitrary weight distributions. In other words, for any number of delegators $m$, we can construct the weight distribution of delegators in such a way that the probability of {\em A} winning post-delegation is the same as the probability of {\em A} winning when a dominant voter exists, and this probability is generally lower than the probability of {\em A} winning when voters follow an {\em EW} distribution, according to Theorems \ref{thm1} and \ref{thm_balance}. 

\begin{example}
    Let $w = [w_j]^{n}_{j=1}$ be the vector of voters' weights before delegation. Let all delegators except one, have weight $\epsilon$ for any arbitrary $\epsilon > 0$. The remained delegator, is more powerful than all others combined, having weight $m\epsilon + \sum_{j\in \ins} w_j$. In post-delegation voting, the voter that received the powerful delegator's delegated votes, has a higher weight than all other voters combined. Therefore, a dominant voter exists in the system, hence the probability of {\em A} winning post-delegation is probability of {\em A} winning when a dominant voter exists. 
\end{example}

At the end of this section, we consider a large election where all voters and delegators have an equal weight. Large elections are modeled using a Poisson random variable with parameter $n$. It is straightforward to see that with a constant number
of delegators, the probability that {\em A} wins converges to one as the total population size
converges to infinity. We show that the same holds even if there are arbitrarily many delegators.
In particular, the result holds even if the number of delegators is a function of $n$.

\begin{proposition}\label{prop: lim_n_to_inf}
$\lim_{n\rightarrow \infty}P_{EW}^d(A\mid \mathbf{1}_m ,n)=1$ for any fixed $p>0.5$ and any $m$, where $m$ can even depend on $n$.
\end{proposition}
\begin{proof}
Let $p > 0.5$ and $\epsilon \in (0, p - \frac{1}{2})$. Define two Poisson random variables: $K$ with parameter $np$, and $L$ with parameter $n(1-p)$. These variables reflect the number of {\em A}-voters and {\em B}-voters, respectively. All results in this proof are asymptotic; to avoid repetition, we omit the phrases "as $n \to \infty$" or "as $m \to \infty$" when it is clear from the context.

The probability that number of {\em A}-voters is at least $n(\frac{1}{2}+\epsilon)$ can be bounded using the Poisson random variable concentration inequalities from Mitzenmacher and Upfal~\cite{Mitzenmacher} which say that for a Poisson random variable $X$ with parameter $\lambda$,

If $x < \lambda$, the following holds:
\begin{equation}
\label{low_concentration}
P(X\leq x)\leq \frac{e^{-\lambda}(e\lambda)^x}{x^x},
\end{equation}
and if  $x > \lambda$, then
\begin{equation}
\label{high_concentration}
P(X\geq x)\leq \frac{e^{-\lambda}(e\lambda)^x}{x^x}.
\end{equation}
For $x=n(\frac{1}{2}+\epsilon)$ and $\lambda_K = np$, we obtain

\begin{equation*}
    P[K<n(\frac{1}{2}+\epsilon)] \leq \frac{(enp)^{n(\frac{1}{2}+\epsilon)}}{e^{np} (n(\frac{1}{2}+\epsilon))^{n(\frac{1}{2}+\epsilon)}} = \left( \frac{\left(\frac{p}{\frac{1}{2}+\epsilon}\right)^{\frac{1}{2}+\epsilon}}{e^{p-\frac{1}{2}-\epsilon}}\right)^n \overset{n \to \infty}{\longrightarrow} 0.
\end{equation*}

The last implication holds if we show that $$\left(\frac{p}{q}\right)^q<e^{p-q},$$
where $q:=\frac{1}{2}+\epsilon$. Consider $f(p)=e^{p-q}-(p/q)^q$. Then, $f(q)=0$ and $f$ is increasing in $p$. The latter holds because $\frac{\partial f(p)}{\partial p}=e^{p-q}-(p/q)^{q-1}>0$ for any $p>q$, since $e^{p-q}>1$ and $(p/q)^{q-1}<1$. 

Hence, the probability that number of {\em A}-voters is at least $n(\frac{1}{2}+\epsilon)$ is the following:
\begin{equation*}
    P[K\geq n(\frac{1}{2}+\epsilon)] = 1-P[K<n(\frac{1}{2}+\epsilon)] \overset{n \to \infty}{\longrightarrow} 1.
\end{equation*}

At the same time, for $x=n(\frac{1}{2}-\epsilon)$ and $\lambda_L = n(1-p)$ we obtain
\begin{equation}\label{eq: l>n(0.5-eps)}
    P[L>n(\frac{1}{2}-\epsilon)] \leq \frac{(en(1-p))^{n(\frac{1}{2}-\epsilon)}}{e^{n(1-p)} (n(\frac{1}{2}-\epsilon))^{n(\frac{1}{2}-\epsilon)}} \overset{n \to \infty}{\longrightarrow} 0.
\end{equation}

The last implication follows from the following: We can rewrite the fraction on the right-hand side as follows:
\begin{align*}
     &\frac{e^{n(\frac{1}{2}-\epsilon)(1+\log(n(1-p)))}}{e^{n(\frac{1}{2}-\epsilon) (\frac{(1-p)}{(\frac{1}{2}-\epsilon)} + \log(n(\frac{1}{2}-\epsilon)))}}  \\ & = \exp\big( n(\frac{1}{2}-\epsilon) \underbrace{\big[ 1+\log(n(1-p)) - \frac{1-p}{\frac{1}{2}-\epsilon} - \log(n(\frac{1}{2}-\epsilon)) \big]}_{=:g} \big).
\end{align*}
Note that $g$ is independent of $n$, as we can write
\begin{equation*}
    g = 1 - \frac{1-p}{\frac{1}{2}-\epsilon} + \log(\frac{1-p}{\frac{1}{2}-\epsilon}).
\end{equation*}
By assumption on $p$ and $\epsilon$, we have that 
\begin{equation*}
     \frac{1-p}{\frac{1}{2}-\epsilon} < 1.
\end{equation*}
We introduce the function  $f(y)$, defined for any $y \in (0,1)$ as follows:
\begin{equation*}
    f(y):= 1 -y+ \log(y).
\end{equation*}
Then $f(\frac{1-p}{\frac{1}{2}-\epsilon}) = g$.
Function $f$ has the following properties: First, $\lim_{y \to 0} f(y) = -\infty$ and $\lim_{y\to 1} f(y) = 0$. Second, the derivative $f'(y) = -1 + \frac{1}{y} > 0$, since $y<1$. Hence, $f(y)$ is negative for any $y<1$. This implies that $g<0$ and hence the right-hand side of equation \eqref{eq: l>n(0.5-eps)} is $\exp(n(\frac{1}{2}-\epsilon)g)$ and converges to $0$ for $n \to \infty$, since $g$ is negative.

Hence, the probability that number of {\em B}-voters is at most $n(\frac{1}{2}-\epsilon)$ is the following: 
\begin{equation*}
    P[L\leq n(\frac{1}{2}-\epsilon)] = 1-P[L > n(\frac{1}{2}-\epsilon)] \overset{n \to \infty}{\longrightarrow} 1.
\end{equation*}

So far, we have shown that the number of {\em A}-voters exceeds $n\left(\frac{1}{2} + \epsilon\right)$ and the number of {\em B}-voters is less than $n\left(\frac{1}{2} - \epsilon\right)$ with high probability. Now, let us focus on $m$ delegators who have independently delegated their votes to random voters. Let $E_K$ and $E_L$ denote the number of votes that {\em A}-voters and {\em B}-voters receive from delegators, respectively. It is evident that if $m < 2\epsilon n$, even if all delegators delegate their vote to {\em B}-voters, {\em A} still wins with high probability since $K - L \geq 2\epsilon n$ with high probability.

Next, we consider the case where $m > 2\epsilon n$. From the above, we know that for sufficiently large $n$, with probability 1, $K \geq n\left(\frac{1}{2} + \epsilon\right)$ and $L \leq n\left(\frac{1}{2} - \epsilon\right)$. Therefore, with high probability,
\begin{align*}
    \frac{K}{L} \geq \frac{\frac{1}{2}+\epsilon}{\frac{1}{2}-\epsilon} > 1.
\end{align*}
Now, let us define the following i.i.d. random variables for each delegator $i \in \{1, \ldots, m\}$:
\begin{equation*}
    X_i : = \begin{cases}
+1 & \text{w.p. } \frac{K}{K+L} \\
-1 & \text{w.p. } \frac{L}{K+L}
\end{cases}
\end{equation*}
$X_i = 1$ if delegator $i$ has delegated their vote to an {\em A}-voter, and $X_i = -1$ otherwise. Define $X = \sum_{i=1}^m X_i$. Here, $X$ represents the surplus of delegated votes to {\em A}-voters; in other words, $X = E_K - E_L$. We know that $\mathbb{E}[X_i] = \frac{K - L}{K + L} \geq 2\epsilon > 0$, and due to the linearity of expectation, $\mathbb{E}[X] \geq 2\epsilon m$.

Hoeffding's inequality states that for independent random variables $X_i$, where each $X_i \in [-1, 1]$, and $X = \sum_{i=1}^m X_i$, $$\Pr[|X - \mathbb{E}[X]| \geq \delta] \leq 2 \exp \left( \frac{-\delta^2}{2m} \right).$$ Using Hoeffding's inequality with $\delta = \epsilon m$,
\begin{equation}
    \begin{aligned}
        P[|X-E|X|| \geq \epsilon m] \leq 2\exp{(\frac{-m\epsilon^2}{2})} \overset{m \to \infty}{\longrightarrow} 0.
    \end{aligned}
\end{equation}

Therefore, $\Pr[X < 0] \overset{m \to \infty}{\longrightarrow} 0$. This implies that, with high probability, $X \geq 0$. The number of votes for voter {\em A} is the sum of $K$ (the original votes from {\em A}-voters) and $E_K$ (the votes obtained from delegators who have delegated their vote to an {\em A}-voter). Similarly, {\em B} has a total of $L + E_L$ votes. The vote surplus for {\em A} is $K + E_K - L - E_L$. We have demonstrated that $X = E_K - E_L \geq 0$, and $K - L \geq 2\epsilon n$ with high probability. Therefore, with high probability, {\em A} receives more votes than {\em B} and wins the election as $n \to \infty$ and $m \to \infty$.
\end{proof}
This result can be generalized to other distributions $F$, by suitably defining the corresponding value of $n$. \\
Theorems \ref{thm1}, \ref{thm_balance}, and Proposition \ref{prop: lim_n_to_inf} lead to the following claim. This claim implies that even a single delegation can substantially alter the winning probability for the ex-ante majority and influence the election outcome.
\begin{claim}
Even a single delegator can reduce the probability of {\em A} winning up to $1-p$.
\end{claim}
\begin{proof}
    Theorems \ref{thm1} and \ref{thm_balance} imply that the maximum probability of {\em A} winning is achieved when all voters have equal weight, and the minimum probability occurs when a dominant voter exists. Note that even a single delegation can shift the weight distribution of voters from an {\em EW} configuration to a {\em DW} configuration. This situation arises when the delegator's weight surpasses the combined weight of all other agents. In such a case, post-delegation, the voter who receives the delegator's votes becomes the dominant voter, solely determining the election's outcome. As \(n\) approaches infinity, proposition \ref{prop: lim_n_to_inf} shows that the probability of {\em A} winning in conventional voting approaches 1. However, post-delegation, this probability is reduced to $p$ due to the presence of a dominant voter. Consequently, this probability change approaches $1-p$ from below.
\end{proof}

\section{Numerical Experiments}\label{sec:numerical}
We conducted a series of 400 experiments to examine the relationship between the probability of \(A\) winning, given the weights of voters, and various measures of imbalance in the distribution of those weights. The voter weight vector \(\omega\) was generated uniformly at random, with each weight ranging from 0 to 1. The number of voters was fixed at 10, i.e., \(n = 10\). In each experiment, we computed the probability of \(A\) winning, denoted as \(P(A \mid \omega)\), when \(p = 0.7\), using equation \eqref{pr_conv}.

To measure the imbalance in each vector \(\omega\), we computed four distinct statistical measures: the Gini coefficient, Variance, the Theil index, and the Hoover index. We then computed the Pearson correlation coefficients between the computed probabilities that {\em A} wins and each of the imbalance measures. The correlations, precise definitions, and mathematical formulas for each measure are detailed below, where $\mu$ denotes the mean of \(\omega\):

\begin{table}[H]
    \centering
    
    \caption{Definitions and Correlations of Inequality Measures with Probability of A Winning.}
    \label{tab:measures_and_correlations}
 \begin{tabular}{|p{2cm}|p{2cm}|p{3.3cm}|p{2.2cm}|}
        \hline
        \textbf{Measure} & \textbf{Definition} & \textbf{Formula} & \textbf{Correlation} \\
        \hline
        \textbf{Gini} & Quantifies inequality of distribution. & 
        $\frac{n + 1 - 2 \sum_{i=1}^{n} \frac{cum(\omega)_i}{cum(\omega)_n}}{n}$ & -0.96 \\
        \hline
        \textbf{Variance} & Measures dispersion of weights. & 
        $\frac{1}{n} \sum_{i=1}^{n} (\omega_i - \mu)^2$ & -0.61 \\
        \hline
        \textbf{Theil} & Evaluates entropy and inequality. & 
        $\frac{1}{n} \sum_{i=1}^{n} \omega_i \log \left( \frac{\omega_i}{\mu} \right)$ & -0.80 \\
        \hline
        \textbf{Hoover} & Portion of weight to redistribute for equality & 
        $\frac{1}{2} \sum_{i=1}^{n} \frac{|\omega_i - \mu|}{\sum_{j=1}^{n} \omega_j}$ & -0.94 \\
        \hline
    \end{tabular}
\end{table}

The correlation matrix from our experiments provide insights into how strongly each measure of imbalance is related to the probability of \(A\) winning. We observe that the balance in the weight distribution of voters highly affects the probability of {\em A} winning. As an example, let us consider the Gini coefficient leading to a correlation of -0.96 with the probability of {\em A} winning. This strong negative correlation indicates that as the voter weight distribution becomes more balanced (lower Gini coefficient), the probability of {\em A} winning increases. These numerical findings extend the findings from Theorems \ref{thm1} and \ref{thm_balance}, demonstrating that delegation favors {\em A} when it promotes a more balanced distribution of voter weights.

The results are visualized through scatter plots in Figures~\ref{fig:gini}--\ref{fig:hoover} showing the relationship between each imbalance measure (on the x-axis) and the computed probability that {\em A} wins (on the y-axis). 

% \begin{figure}
%   \centering
%   \includegraphics[width=0.45\textwidth]{Figures/gini_vs_probability.png}
%   \caption{Gini Coefficient vs Probability}
%   \label{fig:gini}
% \end{figure}

% \begin{figure}
%   \centering
%   \includegraphics[width=0.45\textwidth]{Figures/variance_vs_probability.png}
%   \caption{Variance vs Probability}
%   \label{fig:variance}
% \end{figure}
\begin{figure}[htbp]
  \centering
  \begin{minipage}[b]{0.48\textwidth}
    \centering
    \includegraphics[width=\textwidth]{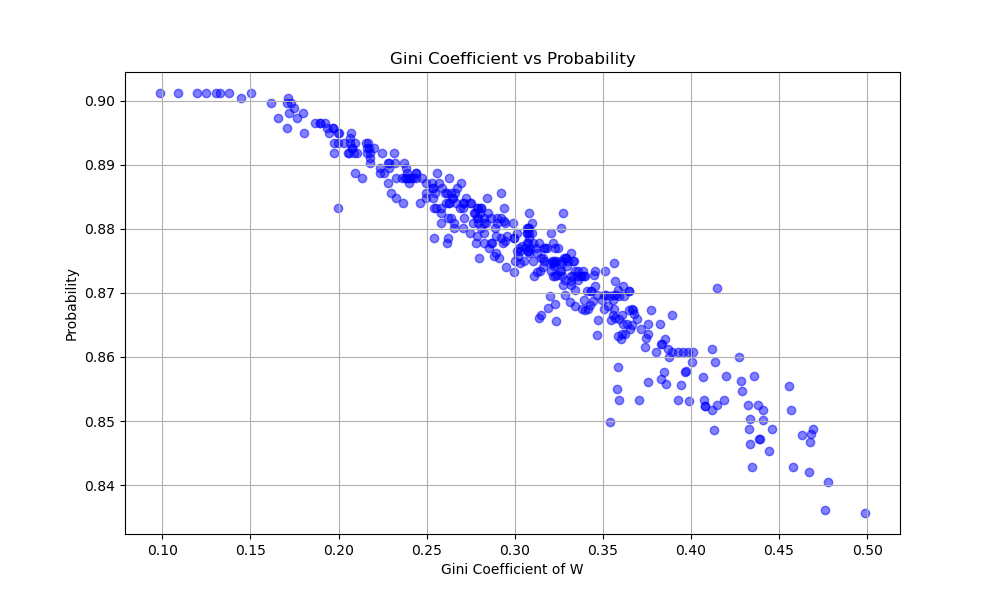}
    \caption{Gini Coefficient vs Probability}
    \label{fig:gini}
  \end{minipage}
  \hfill
  \begin{minipage}[b]{0.48\textwidth}
    \centering
    \includegraphics[width=\textwidth]{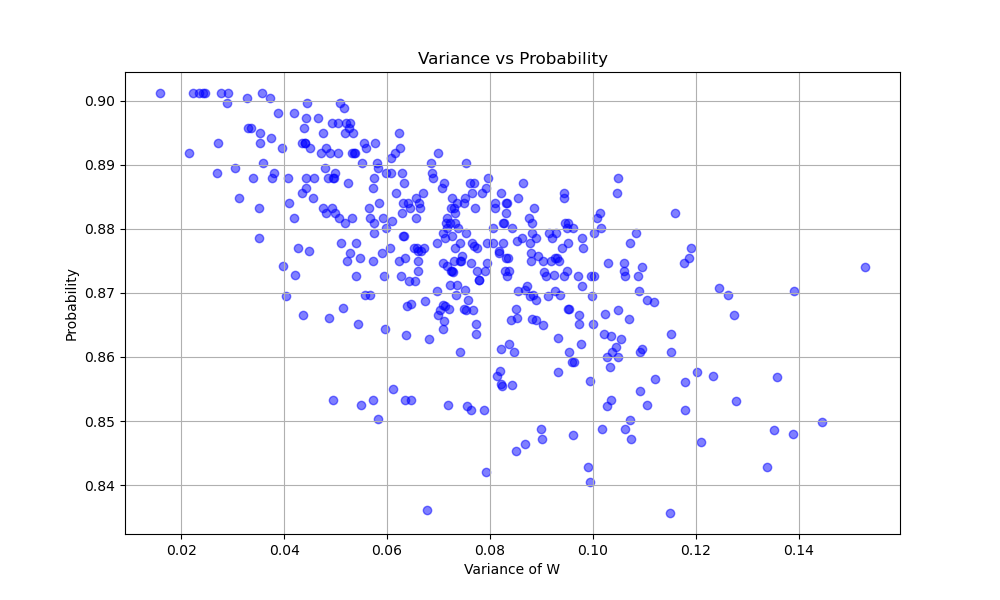}
    \caption{Variance vs Probability}
    \label{fig:variance}
  \end{minipage}
\end{figure}

% \begin{figure}
%   \centering
%   \includegraphics[width=0.45\textwidth]{Figures/theil_vs_probability.png}
%   \caption{Theil Index vs Probability}
%   \label{fig:theil}
% \end{figure}

% \begin{figure}
%   \centering
%   \includegraphics[width=0.45\textwidth]{Figures/hoover_vs_probability.png}
%   \caption{Hoover Index vs Probability}
%   \label{fig:hoover}
% \end{figure}
\begin{figure}[htbp]
  \centering
  \begin{minipage}[b]{0.48\textwidth}
    \centering
    \includegraphics[width=\textwidth]{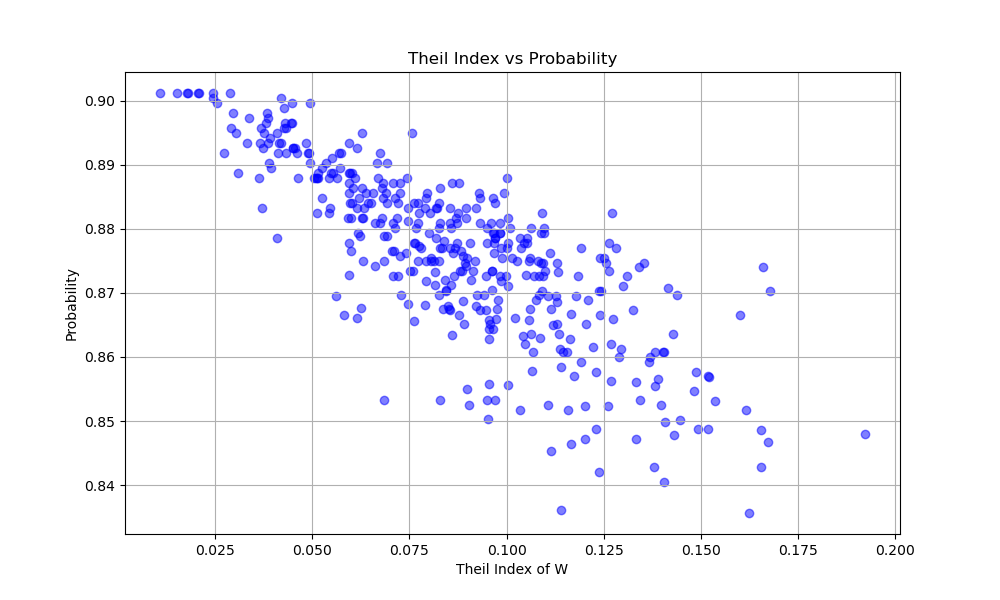}
    \caption{Theil Index vs Probability}
    \label{fig:theil}
  \end{minipage}
  \hfill
  \begin{minipage}[b]{0.48\textwidth}
    \centering
    \includegraphics[width=\textwidth]{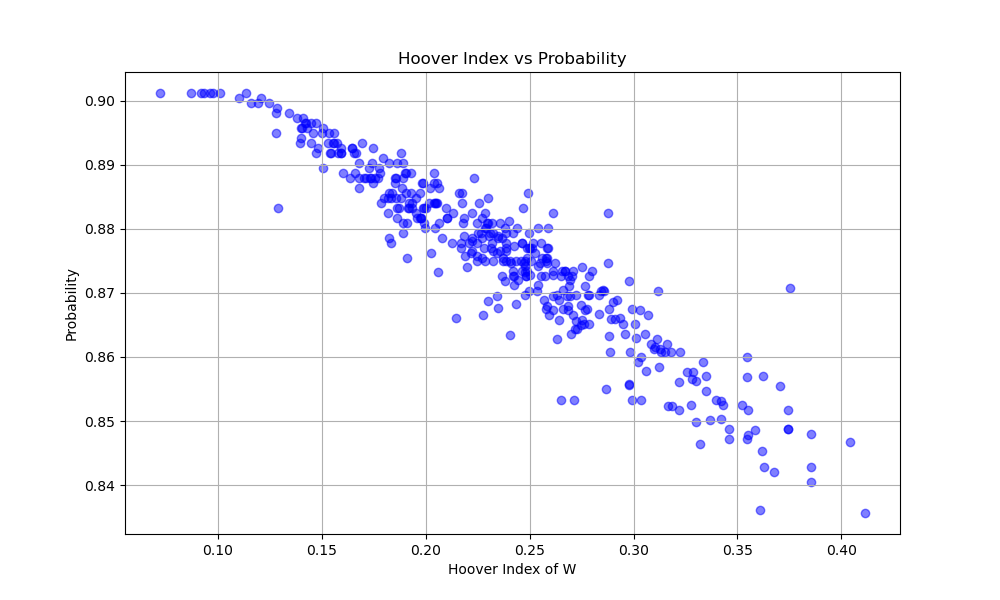}
    \caption{Hoover Index vs Probability}
    \label{fig:hoover}
  \end{minipage}
\end{figure}

\section{Sybil Attacks}\label{sybill-attack}

So far, we have assumed that the identities of agents are fixed and known. However, in blockchains without robust identity issuance mechanisms, it is possible for agents to mimic multiple identities. For instance, they could distribute their tokens across several unconnected wallets. While such actions may involve a cost, this cost can be negligible for agents with substantial stakes. Consequently, such agents may have an incentive to create multiple identities to increase their chances of obtaining votes through delegation. 

There are several ways to prevent Sybil attacks. For instance, one can take a snapshot of all stakes before announcing the vote. This means only wallets that existed at the time of the snapshot—and their balances—are valid for voting. Any newly created Sybil accounts after the announcement won’t have an impact. This also helps reduce vote-buying since people can’t just buy more tokens after the announcement to gain more voting power.
This approach is already used in the blockchain space by platforms like Snapshot. Another method, explored in Project Catalyst (in the Cardano ecosystem), requires wallets to register for voting in advance. The exact voting rules and tallying process are only revealed after registration closes. This uncertainty would make it less appealing for attackers to register Sybil wallets.

Moreover, elections could also require participants to hold their stake for a certain period before being eligible to vote. This way, newly created Sybil accounts wouldn’t qualify since they wouldn’t have held their stake long enough. Finally, we note that Sybil attacks can be mitigated on blockchains with identity issuance mechanisms, where all wallets are linked to unique identities, and voting is restricted to those identities. This highlights a potential rationale for implementing identity issuance procedures in systems with powerful stakeholders, as such mechanisms prevent the dilution of the benefits of delegation caused by Sybil attacks.

Nevertheless, one might be interested in the consequences if such Sybil attacks are possible and happen, and we explore the implications of this phenomenon. One agent could create an unlimited number of accounts by allocating an infinitesimally small stake to each. Since all accounts have an equal probability of receiving delegated votes regardless of stake size, the attacker would, with high probability, capture all delegated votes. Consequently, the final decision would be entirely in their hands. In this scenario, regardless of the initial weight distribution, if the number of delegated votes is sufficiently high, the system effectively reduces to a \textit{DW} scenario, which offers the least benefit to the majority. Therefore, if Sybil attacks cannot be prevented, it is better to prohibit vote delegation in the system.

\section{Conclusions and Open Questions}\label{sec: conclusion}
Vote delegation is now a widespread practice in DeFi, with major protocols using it to increase participation and operational efficiency in decentralized governance.
We have demonstrated that the allowance for vote delegation can significantly alter election outcomes on decentralized systems, in particular blockchains. Crucially, the alternative that benefits from delegation heavily depends on the distribution of voting weights of voters. Our research reveals a striking dichotomy: if voting weights are evenly distributed ({\em EW}), delegation tends to benefit the ex-ante minority. Conversely, if a single voter holds a majority of the voting weights ({\em DW}), it is the ex-ante majority that benefits from delegation. This observation is rooted in the fact that the highest probability of the ex-ante majority winning is obtained when all voters have equal weight, while the lowest probability occurs when a dominant voter exists. Our numerical experiments corroborate this finding and show that this argument can be extended to arbitrary weight distributions. Specifically, if delegation leads to a more balanced distribution of voters' weight, it typically benefits the ex-ante majority. An interesting open question for future research is whether a measure of balance for the weight distribution of voters exists that has a monotonic relation with the probability of ex-ante majority winning.\\

In blockchains, it is typically assumed that the preferred alternative has a higher chance of winning the election. Consequently, the preferred alternative is considered the ex-ante majority, while the non-preferred alternative is regarded as the ex-ante minority. Under this assumption, our findings lead to four main conclusions for governing blockchains:
\begin{enumerate}
    \item For small blockchain communities (DeFi protocols or DAOs) with balanced voting weights, vote delegation is undesirable. Even a single delegator can have a devastating effect, reducing the probability of the preferred alternative winning the election up to $0.5$. Therefore, vote delegation, even with a single delegator, should not be allowed. 
    \item In large communities where all agents have equal weight, vote delegation has no impact as the preferred alternative wins with high probability in either case. In this situation, the decision to allow delegation should be based on other considerations, such as financial incentives.
    \item In blockchains with a highly unbalanced weight distribution of voters, vote delegation helps to balance the weights, enhances the probability of the preferred alternative winning, and thus benefits the entire community.
    \item When a large number of same-weight delegators exist, allowing delegation will shift the weight distribution of voters towards equal weights, making delegation beneficial in this situation.
\end{enumerate}

While our investigation primarily focuses on the impact of vote delegation on the probability of each alternative winning an election, several important aspects remain uncovered. One key area is the strategic behavior of delegators who anticipate the effects of delegation, and adjust their willingness for delegation accordingly. Additionally, examining elections with more than two alternatives presents another intriguing question for future research. Understanding how the distribution of voter weights influences outcomes in multi-alternative elections could provide a broader and more comprehensive understanding of delegation's effects. In addition, empirical validation of our predictions in live DeFi governance systems, such as in Optimism or Compound, could bridge the gap between theoretical models and actual protocol behavior.

\bibliographystyle{alpha}
\newcommand{\etalchar}[1]{$^{#1}$}

% \bibliography{references}

\appendix
% \section{Appendix}\label{app:1}
% In the Appendix we will present all the omitted proofs, and a table of notations.\\ 

% \noindent\textbf{Proof of Theorem \ref{thm1}.}  

% \noindent\textbf{Proof of Proposition \ref{P_lim}.}  

% \noindent\textbf{Proof of Proposition \ref{prop: lim_n_to_inf}.}  

\newpage
\noindent\textbf{Table of Notations}  \\
\begin{center}
\renewcommand{\arraystretch}{1.7}
\begin{tabular}{|l|p{6cm}|}
\hline
\textbf{Notation} & \textbf{Description} \\
\hline
\( n \) & Number of voters \\
\hline
\( m \) & Number of delegators \\
\hline
\( w=[w_j]^n_{j=1} \) & Weight vector of voters \\
\hline
\( D=[d_j]^n_{j=m} \) & Weight vector of delegators \\
\hline
\( w^d=[w^d_j]^n_{j=1} \) & Weight vector of voters post-delegation\\
\hline
\( P(A \mid \omega) \) & Probability of {\em A} winning, when voters weights are given by $\omega$\\
\hline

\( P^d(A\mid D) \) & Probability of {\em A} winning post-delegation\\
\hline

\( P_{DW}^c(A) \) & Probability of {\em A} winning in conventional voting when {\em DW} is the initial distribution of voters \\
\hline

\( P_{DW}^d(A \mid D) \) & Probability of {\em A} winning post-delegation, when {\em DW} is the initial distribution of voters\\
\hline

\( P_{EW}^c(A) \) & Probability of {\em A} winning in conventional voting, when {\em EW} is the initial distribution of voters \\
\hline
\( P_{EW}^d(A \mid D) \) & Probability of {\em A} winning post-delegation, when {\em EW} is the initial distribution of voters \\

\hline
\end{tabular}
\end{center}

\end{document}